\documentclass[a4paper,11pt]{llncs}
\usepackage{ifthen}
\newboolean{pgf}
\setboolean{pgf}{true}
\usepackage[colorlinks=true, citecolor= blue, linkcolor= blue, urlcolor= blue, bookmarks=false]{hyperref}
\usepackage[all]{hypcap}

\usepackage{graphicx,color}      
\usepackage{verbatim}
\usepackage[ruled,linesnumbered]{algorithm2e}
\usepackage{amsmath, amsfonts, amssymb}
\usepackage[a4paper,margin=3.72cm,centering]{geometry}
\usepackage{subfig}
\usepackage{ifthen}
\usepackage{url}
\spnewtheorem{fact}{Fact}{\bf}{\it}

\newboolean{includeXXX}
\newboolean{includecut}

\setboolean{includeXXX}{false}
\setboolean{includecut}{false}

\ifthenelse{\boolean{includecut}}{
\newcommand{\cut}[1]{\textcolor{blue}{#1}}

\newcommand{\Rcut}[1]{}  

}
{
\newcommand{\cut}[1]{}

\newcommand{\Rcut}[1]{#1} 

}

\ifthenelse{\boolean{includeXXX}}{
\newcommand{\XXX}[1]{\textcolor{red}{XXX #1 XXX\ }}

}
{
\newcommand{\XXX}[1]{}

}

\newcommand{\gap}[1]{} 

\usepackage{tikz}
\usetikzlibrary{scopes}
\usetikzlibrary{arrows,shapes,calc,shapes,patterns,backgrounds}

\newcommand{\cost}{\text{cost}}

\usepackage{xstring}

\newboolean{colour}
\setboolean{colour}{true}

\ifthenelse{\boolean{pgf}}{
\ifthenelse{\boolean{colour}}{
\tikzstyle{c2}=[color=blue!70]
\tikzstyle{c1}=[color=red!70]
\tikzstyle{c3}=[color=orange!70]
\tikzstyle{c6}=[color=pink!70]
\tikzstyle{c5}=[color=cyan!70]
\tikzstyle{c4}=[color=green!70]

}{
\tikzstyle{c1}=[color=black!20]
\tikzstyle{c2}=[color=black!50]
\tikzstyle{c3}=[color=black!75]
\tikzstyle{c4}=[color=gray!70]
\tikzstyle{c5}=[color=black]
\tikzstyle{c6}=[pattern=north west lines]

}

\tikzstyle{c7}=[color=gray!70]
\tikzstyle{c8}=[color=white]
\tikzstyle{ct8}=[color=white]

\newcommand{\drawbox}[2]{\draw[very thick,color=black] (0,0) rectangle (#1,-#2);}

\newcommand{\drawrow}[2]{

		\StrLen{#2}[\l] 

		\scope[shift={(0,-#1)}]

		\foreach \x in {1,2,...,\l} 
			{
				\StrChar{#2}{\x}[\ch]

				\if \ch0

				\else
					\draw[fill,c\ch] (\x-1,0) rectangle +(1,1);
					\draw[color=black!80] (\x-1,0) rectangle +(1,1);
				\fi

			}

		\endscope
}

\newcommand{\drawrowL}[2]{

		\StrLen{#2}[\l] 

		\scope[shift={(0,-#1)}]

		\foreach \x in {1,2,...,\l} 
			{
				\StrChar{#2}{\x}[\ch]

				\if \ch0

				\else
					\draw (\x-0.5,0.37) node[anchor=mid,color=white] {{\bf \ch}};
				\fi

			}

		\endscope
}

}{}

\renewcommand{\geq}{\geqslant}
\renewcommand{\leq}{\leqslant}

\newcommand{\Ptime}{\ensuremath{\mathbf{P}}}
\newcommand{\NPtime}{\ensuremath{\mathbf{NP}}}
\newcommand{\PSPACE}{\ensuremath{\mathbf{PSPACE}}}

\newcommand{\Coloroid}[1]{$#1$-\textsc{Flood-It}}
\newcommand{\ColoroidFree}[1]{$#1$-\textsc{Free-Flood-It}}
\newcommand{\ColoroidF}{\textsc{Free-Flood-It}}
\newcommand{\Floodit}{\textsc{Flood-It}}
\newcommand{\Mnc}{\ensuremath{\max\{m(B)\,|\,B \in B_{n,c}\}}}
\newcommand{\SCS}{\textsc{SCS}}

\newcommand{\recdim}[2]{$#1$$\mspace{1mu}$$\times$$\mspace{1mu}$$#2$}

\title{The Complexity of Flood Filling Games}
\author{Rapha\"{e}l~Clifford \and  Markus~Jalsenius \and Ashley~Montanaro \and Benjamin~Sach}
\institute{Department of Computer Science, University of Bristol, UK}

\begin{document}

\maketitle

\begin{abstract}
We study the complexity of the popular one player combinatorial game known as Flood-It. In this game the player is given an \recdim{n}{n} board of tiles where each tile is allocated one of $c$ colours. The goal is to make the colours of all tiles equal via the shortest possible sequence of flooding operations. In the standard version, a flooding operation consists of the player choosing a colour~$k$, which then changes the colour of all the tiles in the monochromatic region connected to the top left tile to $k$. After this operation has been performed, neighbouring regions which are already of the chosen colour $k$ will then also become connected, thereby extending the monochromatic region of the board. We show that finding the minimum number of flooding operations is \NPtime-hard for $c \geq 3$ and that this even holds when the player can perform flooding operations from any position on the board. However, we show that this `free' variant is in \Ptime\ for $c=2$. We also prove that for an unbounded number of colours, Flood-It remains \NPtime-hard for boards of height at least 3, but is in \Ptime\ for boards of height 2. Next we show how a $(c-1)$ approximation and a randomised $2c/3$ approximation algorithm can be derived, and that no polynomial time constant factor, independent of $c$, approximation algorithm exists unless \Ptime=\NPtime.  We then investigate how many moves are required for the `most demanding' \recdim{n}{n} boards (those requiring the most moves) and show that the number grows as fast as $\Theta(\sqrt{c}\, n)$.    Finally, we consider boards where the colours of the tiles are chosen at random and show that for $c\geq 2$, the number of moves required to flood the whole board is $\Omega(n)$ with high probability.
\end{abstract}

\section{Introduction} \label{sec:intro}
In the popular one player combinatorial game known as Flood-It, each tile of an \recdim{n}{n} board is allocated one of $c$ colours, where $c$ is a parameter of the game.  Two left/right/up/down adjacent tiles are said to be connected if they have the same colour and a (connected) region of the board is defined to be any maximal connected component.   The standard version of the game starts with the player `flooding' the region that contains the top left tile. The flooding operation simply involves changing the colour of all the tiles in the region to be some new colour. However, this also has the effect of connecting the newly flooded region to all neighbouring regions of this colour.    The overall aim is to flood the entire board, that is connect all regions, in as few flooding operations as possible. Every flooding operation changes the colour of the region that contains the top left tile.  Figure~\ref{fig:turns} gives an example of the first few moves of a game. The border shows the outline of the region which has so far been flooded.
\begin{figure}[t]
    \centering
	\tikzstyle{box}=[very thick]
	\begin{tikzpicture}[inner sep=2pt,scale=0.255]

	\scope[shift={(0,0)}]

		\drawrow{1}{131213}
		\drawrow{2}{212323}
		\drawrow{3}{223222}
		\drawrow{4}{231213}
		\drawrow{5}{231333}
		\drawrow{6}{121321}

		\draw[box] (0,0) rectangle +(6,-6);
		\draw[box] (0,0) rectangle +(1,-1);
	\endscope

	\fill[c2] (7,-3) rectangle +(1,1);
	\draw[] (7,-3) rectangle +(1,1);

	\draw[-open triangle 60] (6.5,-4) -- +(2,0);
	\scope[shift={(9,0)}]
		\drawrow{1}{231213}
		\drawrow{2}{212323}
		\drawrow{3}{223222}
		\drawrow{4}{231213}
		\drawrow{5}{231333}
		\drawrow{6}{121321}

		\draw[box] (0,0) rectangle +(6,-6);
		\draw[box] (1,0) -- ++(0,-2) -- ++(1,0) --++ (0,-1) --++ (-1,0) --++ (0,-2) --++(-1,0) ;
	\endscope

	\fill[c3] (16,-3) rectangle +(1,1);
	\draw[] (16,-3) rectangle +(1,1);

	\draw[-open triangle 60] (15.5,-4) -- +(2,0);

	\scope[shift={(18,0)}]
		\drawrow{1}{331213}
		\drawrow{2}{312323}
		\drawrow{3}{333222}
		\drawrow{4}{331213}
		\drawrow{5}{331333}
		\drawrow{6}{121321}

		\draw[box] (0,0) rectangle +(6,-6);
		\draw[box] (2,0) -- ++(0,-1) --++(-1,0)  -- ++(0,-1) --++(2,0) --++ (0,-1) -- ++(-1,0) --++ (0,-2) --++ (-2,0);
	\endscope

	\fill[c1] (25,-3) rectangle +(1,1);
	\draw[] (25,-3) rectangle +(1,1);

	\draw[-open triangle 60] (24.5,-4) -- +(2,0);

	\scope[shift={(27,0)}]
		\drawrow{1}{111213}
		\drawrow{2}{112323}
		\drawrow{3}{111222}
		\drawrow{4}{111213}
		\drawrow{5}{111333}
		\drawrow{6}{121321}

		\draw[box] (0,0) rectangle +(6,-6);
		\draw[box] (3,0) -- ++(0,-1) -- ++(-1,0)  -- ++(0,-1) -- ++(1,0) -- ++ (0,-4) -- ++(-1,0) -- ++(0,1) -- ++(-1,0) -- ++ (0,-1);
	\endscope

	\fill[c2] (34,-3) rectangle +(1,1);
	\draw[] (34,-3) rectangle +(1,1);

	\draw[-open triangle 60] (33.5,-4) -- +(2,0);

	\scope[shift={(36,0)}]
		\drawrow{1}{222213}
		\drawrow{2}{222323}
		\drawrow{3}{222222}
		\drawrow{4}{222213}
		\drawrow{5}{222333}
		\drawrow{6}{222321}

		\draw[box] (0,0) rectangle +(6,-6);
		\draw[box] (4,0) -- ++(0,-1) -- ++(-1,0)  -- ++(0,-1) -- ++(1,0) -- ++ (0,1) -- ++(1,0) -- ++(0,-1) -- ++(1,0) -- ++ (0,-1) -- ++(-2,0) -- ++(0,-1) -- ++(-1,0) -- ++(0,-2);
	\endscope

	\end{tikzpicture}  
    \caption{A sequence of four moves on a \recdim{6}{6} Flood-It board with 3 colours. }
    \label{fig:turns}
\end{figure}
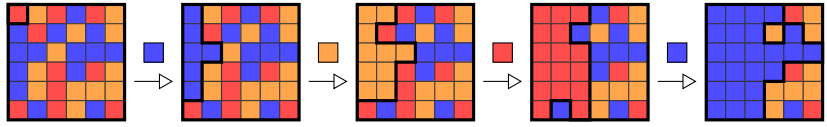

In this paper, we investigate a number of questions inspired by Flood-It. We first show that not only are natural greedy approaches to the game bad, but in fact finding an optimal solution (one which requires the fewest possible moves) for Flood-It is \NPtime-hard for $c \geq 3$, and that this also holds for a variant of the game we call Free-Flood-It where the player can perform flooding operations at any position on the board. On the other hand, we show that solving Free-Flood-It with $c=2$ is in \Ptime.  We also consider the effect of changing the shape of the board, and prove that Flood-It remains \NPtime-hard for rectangular boards of height at least 3, with an unbounded number of colours, but is in \Ptime\ for boards of height 2. As a stepping stone, we also prove \NPtime-hardness of a restricted version of the well-studied shortest common supersequence problem (q.v.).

 Next we show how a $(c-1)$ approximation and a randomised $2c/3$ approximation algorithm for Flood-It can be derived. However, no polynomial time constant factor, independent of $c$, approximation algorithm exists unless \Ptime=\NPtime.  We then consider how many moves are required for the most demanding boards and show that the number grows as fast as $\Theta(\sqrt{c}\, n)$. We say that a board is one of the most demanding boards if it requires at least as many moves as any other board which has the same size and number of colours. Finally, we investigate boards where the colours of the tiles are chosen at random and give a simple proof that for $c\geq 3$, the number of moves required to flood the whole board is $\Omega(n)$ with high probability. We then observe that the same result can in fact be proven for $c \geq 2$ by appealing to previous deep results in percolation theory~\cite{CW1993:random,FN1993:random}; indeed, our work can be seen as a drastic simplification of these results for the case $c \geq 3$.



\vskip 3pt
\noindent {\em History and related work:} Perhaps the most famous recent hardness result involving a popular game is the \NPtime-completeness of Tetris~\cite{DHL2003:Tetris}. Flood-It seems to be a somewhat newer game than Tetris, first making its appearance online in early 2006 courtesy of a company called Lab Pixies.  Since then numerous versions have become available for almost every conceivable platform. We have very recently become aware of a sketch proof by Elad Verbin posted on a blog of the \NPtime-hardness of Flood-It with $6$ colours~\cite{Verbinblog:2009}. Although our work was completed independently, it is interesting to note that there is some similarity to the techniques used in our \NPtime-hardness proof for $c\geq3$ colours.

Independently of this work, Fleischer and Woeginger have studied a closely related game to Flood-It, known as Honey-Bee \cite{FW2010:honey}. This game is also based around repeatedly applying a flood filling operation on a grid. The main differences are that the grid is hexagonal and may contain barriers, and also that there is a two-player variant of the game. In this variant, two players start flood filling from opposite corners, and the goal is to control more of the board than your opponent. Fleischer and Woeginger focus on the computational complexity of Honey-Bee, and consider a number of generalisations of the single player game to different classes of graphs. They prove that some generalisations are \NPtime-hard, while others are in \Ptime. Again, there is some similarity in the techniques used in one of their \NPtime-hardness proofs, although we note that this proof does not immediately apply to Flood-It without some modification. Fleischer and Woeginger also show that the two-player game on arbitrary graphs is \PSPACE-complete.

Another related game whose computational complexity has been studied in detail is known as Clickomania~\cite{BDDFJM2002:Click}.  A rectangular board is initialised in the same way as in Flood-It. The move permitted is for the player to remove a chosen connected monochromatic component of at least two tiles after which any blocks above it will fall down as far as they can. Finding an optimal solution to Clickomania has been shown to be \NPtime-hard for two or more columns and five or more colours, or five or more columns and three or more colours.

There is also existing work on a majority-based recolouring game on graphs \cite{Berger:2001,Flocchini:2003,Peleg:1998}. The game is played over a number of rounds on a simple undirected graph where each vertex is initially coloured white or black. In each round each vertex is recoloured by the colour of the majority of its neighbours. The player's only interaction is to determine the set of vertices which are initially coloured white. The goal is to pick the smallest possible set of vertices such that after a finite number of rounds, all vertices are white.

Flood-It can be thought of as a model for a number of different (possibly not entirely) real world applications.  For example, our results supplement that of recent work on zombie infestation~\cite{MHIS2009:Zombies} if one regards the flooding operation as one where the minds of neighbouring non-zombies are infected by those who have already been turned into zombies.  A separate but no less significant line of research considers the complexity of tools commonly provided with Microsoft Windows. Previous work has shown that aspects of Excel~\cite{IMO2009:Draw} and even Minesweeper~\cite{Kaye2000:Minesweeper} are \NPtime-complete. Our work extends this line of research by showing that flood filling in Microsoft's Paint application is also \NPtime-hard.

\subsection{Notation and definitions}

Let $B_{n,c}$ be the set of all \recdim{n}{n} boards with at most $c$ colours. We write $m(B)$ for the minimum number of moves required to flood a board $B \in B_{n,c}$. We will refer to rows and columns in a board in the usual manner. We further denote the colour of the tile in row $i$ and column $j$ as $B[i,j]$; colours are represented by integers between 1 and $c$. Throughout we assume that $2 \leq c \leq n^2$.
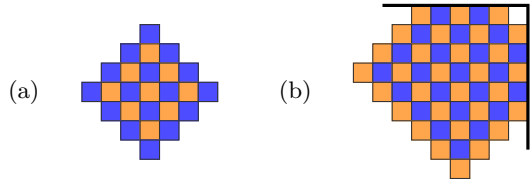
\begin{figure}[t] \centering

        \centering
        \tikzstyle{box}=[thick]
	\begin{tikzpicture}[inner sep=2pt,scale=0.255]
	\scope[shift={(0,0)}]

		\drawrow{1}{0002000}
		\drawrow{2}{0023200}
		\drawrow{3}{0232320}
		\drawrow{4}{2323232}
		\drawrow{5}{0232320}
		\drawrow{6}{0023200}
		\drawrow{7}{0002000}
 		\draw (-3,-3.5) node {(a)};
 		\draw (11,-3.5) node {(b)};

	\endscope

	\scope[shift={(14,3)}]

 		\drawrow{3} {000323230}
		\drawrow{4} {003232323}
		\drawrow{5} {032323232}
		\drawrow{6} {323232323}
		\drawrow{7} {032323232}
		\drawrow{8} {003232323}
		\drawrow{9} {000323230}
		\drawrow{10}{000032300}
		\drawrow{11}{000003000}
	
		\draw[very thick] (1.5,-2) -- ++(7.5,0) -- ++ (0,-7.5);
	\endscope

	\end{tikzpicture} 
        \caption{(a) An alternating 4-diamond and (b) a cropped 6-diamond.\label{fig:kdiamonds}}
\end{figure}

We define a \emph{diamond} to be a diamond-shaped subset of the board (see Figure~\ref{fig:kdiamonds}a). These structures are used throughout the paper. The centre of the diamond is a single tile and the \emph{radius}  is the number of tiles from its centre to its leftmost tile. We write $r$-diamond to denote a diamond of radius $r$. A single tile is therefore a 1-diamond. For $i\in \{1,\dots,r\}$, the $i$th \emph{layer} of an $r$-diamond is the set of tiles at board distance $i-1$ from its centre. We will also consider diamonds which are cropped by intersection with the board edges as in Figure~\ref{fig:kdiamonds}b.


\section{A greedy approach is bad} \label{sec:greed}
An obvious strategy for playing the Flood-It game is the greedy approach. There are two natural greedy algorithms: (1)~we pick the colour that results in the largest gain (number of acquired tiles), or (2)~we choose the colour dominating the perimeter of the currently flooded region. It turns out that both these approaches can be surprisingly bad.

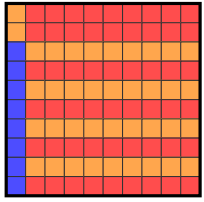
\begin{figure}[t] \centering
        \centering
        \tikzstyle{box}=[thick]
	\begin{tikzpicture}[inner sep=2pt,scale=0.255]
	\scope[shift={(0,0)}]

		\drawrow{1} {3111111111}
		\drawrow{2} {3111111111}
		\drawrow{3} {2333333333}
		\drawrow{4} {2111111111}
		\drawrow{5} {2333333333}
		\drawrow{6} {2111111111}
		\drawrow{7} {2333333333}
		\drawrow{8} {2111111111}
		\drawrow{9} {2333333333}
		\drawrow{10}{2111111111}

		\drawbox{10}{10};
	

	\endscope

%
%
%
%
%
%

	\end{tikzpicture} 
        \caption{A \recdim{10}{10} board where a greedy approach is bad.\label{fig:greedy}}
\end{figure}

To see this, let $B$ be the \recdim{10}{10} board on three colours illustrated in Figure~\ref{fig:greedy}.   The number of moves required to flood $B$ is three.  However, either greedy approach given would first pick the colours appearing on the horizontal lines before finally choosing to flood the left-hand vertical column. In both cases, this requires $10$ moves to fill the board.   It should be clear how this example can easily be extended to arbitrarily large \recdim{n}{n} boards. In general, the greedy algorithm will make $n$ moves, while the optimal algorithm will still make only three.

\section{The complexity of Flood-It} \label{sec:NPC}
\newcommand{\stl}{w}


Let \Coloroid{c} denote the problem which takes as input an \recdim{n}{n} board $B$ of $c$ colours and outputs the minimum number of moves $m(B)$ in a Flood-It game that are required to flood $B$. Similarly, let \ColoroidFree{c} denote the generalised version of \Coloroid{c} in which we are free to flood fill from an arbitrary tile in each move. Although we have seen that a straightforward greedy algorithm fails, it is not too far-fetched to think that a dynamic programming approach would solve these problems efficiently, but the longer one ponders over it, the more inconceivable it seems. To aid frustrated Flood-It enthusiasts, we prove in this section that both \Coloroid{c} and \ColoroidFree{c} are indeed \NPtime-hard, even when the number of colours is as small as three. Interestingly, we will see that \ColoroidFree{2} is in \Ptime.

To show \NPtime-hardness, we reduce from the \emph{shortest common supersequence} problem, denoted \SCS{}, which is defined as follows. The input is a set $S$ of $k$ strings over an alphabet $\Sigma$. A \emph{common supersequence} $s$ of the strings in $S$ is a string such that every string in $S$ is a subsequence of $s$. The output is the length of a shortest common supersequence of the strings in $S$. The decision version of \SCS{} takes an additional integer $\ell$ and outputs yes if the shortest common supersequence has length at most $\ell$, otherwise it outputs no.

Maier~\cite{Mai1978:SCS} showed in 1978 that the decision version of \SCS{} is \NPtime-complete if the alphabet size $|\Sigma|\geq 5$. A couple of years later, R\"aih\"a and Ukkonen~\cite{RU1981:SCS} extended this result to hold for $|\Sigma|\geq 2$. For a long time, various groups of people tried to approximate \SCS{} but no polynomial-time algorithm with guaranteed approximation bound was to be found. It was not until 1995 that Jiang and Li~\cite{JM1995:SCS} settled this open problem by proving that no polynomial-time algorithm can achieve a constant approximation ratio for \SCS{}, unless $\Ptime=\NPtime$. Their result holds for an unbounded alphabet.

The following lemma proves the \NPtime-hardness of both \Coloroid{c} and \ColoroidFree{c} when the number of colours is at least four. The inapproximability of both problems follows immediately from the approximation preserving nature of the reduction. However, in the reduction we present, the number of colours in the \Coloroid{c} instance will be exactly twice the number of alphabet symbols in the \SCS{} instance. For this reason, our inapproximability results only hold when the number of colours is unbounded. We will need a more specialised reduction for the case $c=3$, which is given in Lemma~\ref{lem:NPC-three}.

\begin{lemma}
    \label{lem:NPC-four}
    For $c\geq 4$, \Coloroid{c} and \ColoroidFree{c} are \NPtime-hard (and the decision versions are \NPtime-complete). Further, for an unbounded number of colours $c$, there is no polynomial-time constant factor approximation algorithm, unless $\Ptime=\NPtime$.
\end{lemma}
\begin{proof}
    The proof is split into two parts; first we prove the lemma for \Coloroid{c} in which we flood fill from the top left tile in each move, and in the second part we generalise the proof to \ColoroidFree{c} in which we can flood fill from any tile in each move.

    We reduce from an instance of \SCS{} that contains $k$ strings $s_1,\dots,s_k$ each of length at most $\stl$ over the alphabet $\Sigma$. Suppose that $\Sigma=\{a_1,\dots,a_r\}$ contains $r\geq 2$ letters and let $\Sigma'=\{b_1,\dots,b_r\}$ be an alphabet with $r$ new letters. For $i\in \{1,\dots,k\}$, let $s'_i$ be the string obtained from $s_i$ by inserting the character $b_j$ after each $a_j$ and inserting the character $b_1$ at the very front. For example, from the string $a_3a_1a_4a_3$ we get $b_1a_3b_3a_1b_1a_4b_4a_3b_3$.

    Let $\Sigma\cup \Sigma'$ represent the set of $2r$ colours that we will use to construct a board $B$. First, for $i\in \{1,\dots,k\}$, we define the $|s'_i|$-diamond $D_i$ such that the $j$th layer will contain only one colour which will be the $j$th character from the right-hand end of $s'_i$. Thus, the colour of the outermost layer of $D_i$ is the first character of $s'_i$ (which is $b_1$ for all strings) and the centre of $D_i$ is the last character of $s'_i$. The reason why we intersperse the strings with letters from the auxiliary alphabet $\Sigma'$ is to ensure that no two adjacent layers of a diamond have the same colour. This property is crucial in our proof. Let $B$ be a sufficiently large \recdim{n}{n} board constructed by first colouring the whole board with the colour $b_1$ and then placing the $k$ diamonds $D_i$ on $B$ such that no two diamonds overlap. Since each of the $k$ diamonds has a radius of at most $2\stl+1$, we can be assured that $n$ never has to be greater than $k(4\stl+1)$.

    Suppose that $s$ is a shortest common supersequence of $s_1,\dots,s_k$ and suppose its length is $\ell$. We will now argue that the minimum number of moves to flood $B$ is exactly $2\ell$, first showing that $2\ell$ moves are sufficient. Let $s'$ be the $2\ell$-long string obtained from $s$ by inserting the character $b_j$ after each $a_j$. We make $2\ell$ moves by choosing the colours in the same order as they appear in $s'$. Note that we flood fill from the top left tile in each move. From the construction of the diamonds~$D_i$ it follows that all diamonds, and hence the whole board, are flooded after the last character of $s'$ has been processed.

    It remains to be shown that at least $2\ell$ moves are necessary to flood $B$. Let $s''$ be a string over the alphabet $\Sigma\cup \Sigma'$ that specifies a shortest sequence of moves that would flood the whole board $B$. From the construction of the diamonds $D_i$ it follows that the string obtained from $s''$ by removing every character in $\Sigma'$ is a common supersequence of $s_1,\dots,s_k$ and therefore has length at least $\ell$. By symmetry (replace every $a_j$ with $b_j$ in the strings $s_1,\dots,s_k$), the string obtained from $s''$ by removing every character in $\Sigma$ has length at least $\ell$ as well. Thus, the length of $s''$ is at least $2\ell$.

    Since the decision version of \SCS{} is \NPtime-complete even for a binary alphabet $\Sigma$, it follows that \Coloroid{c} is \NPtime-hard for $c\geq 4$, and the decision version is \NPtime-complete.  As discussed above, observe that the number of colours used is exactly twice the alphabet size of the \SCS{} instance. Therefore the inapproximability result for an unbounded number of colours in the statement of the lemma follows immediately from the approximation preserving nature of the reduction given.

    Now we show how to extend these results to \ColoroidFree{c}. The reduction from \SCS{} is similar to the previously presented reduction. However, instead of constructing only one board $B$, we construct $2k\stl+1$ copies of $B$ and put them together to one large \recdim{n'}{n'} board $B'$. If necessary in order to make $B'$ a square, we add sufficiently many \recdim{n}{n} boards that are filled only with the colour $b_1$. Note that $(2k\stl+1)n$ and hence $(2k\stl+1)k(4\stl+1)$ is a generous upper bound on $n'$.

    From the construction of $B'$ it follows that exactly $2\ell$ moves are required to flood $B'$ if we flood fill from the top left tile in each move; all copies of $B$ will be flooded simultaneously. The question is whether we can do better by flood filling from tiles other than the top left one (or any tile in its connected component). That is, can we do better by picking a tile inside one of the diamonds? We will argue that the answer is no. First note that $2\ell\leq 2kw$. Suppose that we do flood fill from a tile inside some diamond $D$ for some move. This move will clearly not affect any of the other diamonds on $B'$. Suppose that this move would miraculously flood the whole of $D$ in one go so that we can disregard it in the subsequent moves. However, there were originally $2kw+1$ copies of $D$, which is one more than the absolute maximum number of moves required to flood $B'$, hence we can use a recursive argument to conclude that flood filling from a tile inside a diamond will do us no good and would only result in more moves than if we choose to flood fill from the top left tile in each move.
    \qed
\end{proof}

The reduction in the previous proof is approximation preserving, which allowed us to prove that there is no efficient constant factor approximation algorithm. We reduced from an instance of \SCS{} by doubling the alphabet size, resulting in instances of \Coloroid{c} and \ColoroidFree{c} with $c\geq 4$ colours. To establish \NPtime-hardness for $c=3$ colours, we need to consider a different reduction. We do this in the lemma below by reducing from the decision version of \SCS{} over a binary alphabet to the decision versions of \Coloroid{3} and \ColoroidFree{3}. This reduction is not approximation preserving as in the previous proof; the number of moves required to flood the board in the reduced instance of \Coloroid{3} (or \ColoroidFree{3}) does not correspond in a straightforward way to the length of shortest common supersequence in the \SCS{} instance we reduce from.

\begin{lemma}
    \label{lem:NPC-three}
    \Coloroid{3} and \ColoroidFree{3} are \NPtime-hard (and the decision versions are \NPtime-complete).
\end{lemma}
\begin{proof}
    We reduce from an instance of the decision version of \SCS{} on $k$ strings $s_1,\dots,s_k$ of length at most $\stl$ over the binary alphabet $\{1,2\}$ and an integer $\ell$. The yes/no question is whether there exists a common supersequence of length at most $\ell$.

    For $i\in \{1,\dots,k\}$, let $s'_i$ be the string obtained from $s_i$ by inserting the new character~3 at the front of $s_i$ and after each character of $s_i$. Let the set $\{1,2,3\}$ represent the colours that we will use to construct a board $B$. First, for each of the $k$ strings $s'_i$ we define the diamond $D_i$ exactly as in the proof of Lemma~\ref{lem:NPC-four} (see Figure~\ref{fig:NPC}a). We define $R$ to be the following rectangular area of the board of width $4\ell+5$ and height $2\ell+3$. Let $x$ be the middle tile at the bottom of $R$. Around $x$ we have layers of concentric half rectangles (see Figure~\ref{fig:NPC}b). We refer to these layers as \emph{arches}, with the first arch being $x$ itself. As demonstrated in the figure, the first arch has the colour~1 and the second arch has the colour~2. All the remaining odd arches have the colour~3, and all the remaining even arches are coloured~2 everywhere except for the tile above $x$ which has the colour~1. As described in detail below, the purpose of these arches is to control which minimal sequences of moves would flood $B$.

\begin{figure}[t]
\centering
	\tikzstyle{box}=[thick]
	\begin{tikzpicture}[inner sep=2pt,scale=0.255]

	\scope[shift={(0,0)}]
		\draw (2,-0.5) node {(a)};
 		\drawrow{0} {00000000200000000}	
 		\drawrow{1} {00000002120000000}	
 		\drawrow{2} {00000021212000000}
 		\drawrow{3} {00000212321200000}
 		\drawrow{4} {00002123232120000}
 		\drawrow{5} {00021232123212000}
		\drawrow{6} {00212321212321200}
		\drawrow{7} {02123212121232120}
		\drawrow{8} {21232121212123212}
		\drawrow{9} {02123212121232120}
		\drawrow{10}{00212321212321200}
		\drawrow{11}{00021232123212000}
		\drawrow{12}{00002123232120000}
		\drawrow{13}{00000212321200000}
 		\drawrow{14}{00000021212000000}	
 		\drawrow{15}{00000002120000000}	
 		\drawrow{16}{00000000200000000}	

	\endscope

	\scope[shift={(20,-5)}]
		\draw (10.5,1.5) node {(b)};
		\drawrow{1} {222222222222222222222}
		\drawrow{2} {211111111131111111112}
		\drawrow{3} {212222222222222222212}
		\drawrow{4} {212111111131111111212}
		\drawrow{5} {212122222222222221212}
		\drawrow{6} {212121111131111121212}
		\drawrow{7} {212121222222222121212}
		\drawrow{8} {212121211131112121212}
		\drawrow{9} {212121212222212121212}
		\drawrow{10}{212121212111212121212}
		\drawrow{11}{212121212131212121212}

	\fill[c3] (-5,4) rectangle +(1,1);
	\draw[] (-5,4) rectangle +(1,1);
	\draw (0.5,4.5) node {= Colour 1};

	\fill[c1] (5,4) rectangle +(1,1);
	\draw[] (5,4) rectangle +(1,1);
	\draw (10.5,4.5) node {= Colour 2};

	\fill[c2] (15,4) rectangle +(1,1);
	\draw[] (15,4) rectangle +(1,1);
	\draw (20.5,4.5) node {= Colour 3};

	\endscope

	\end{tikzpicture}
\caption{ An example of (a) a diamond and (b) a rectangle constructed in the proof of Lemma~\ref{lem:NPC-three}.}
\label{fig:NPC}
\end{figure}
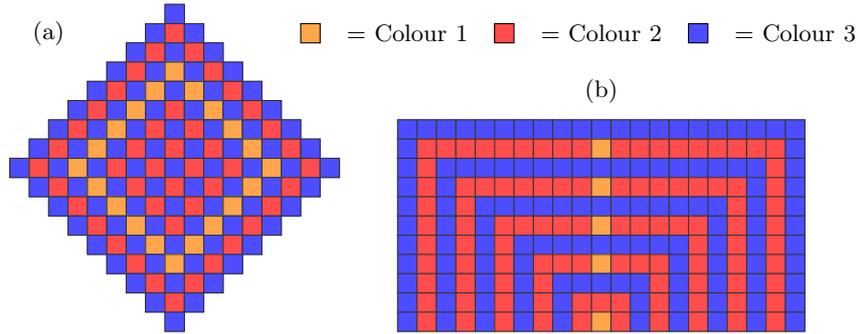

    Let $B$ be a sufficiently large \recdim{n}{n} board constructed as follows. First colour the whole board with the colour~3. Then, at the bottom of $B$ starting from the left, place $2\ell+3$ copies of $R$ one after another without any overlaps. Finally place the $k$ diamonds $D_i$ on $B$ such that no two diamonds overlap and no diamond overlaps any copy of $R$. Figure~\ref{fig:NPC2} illustrates a board $B$ with $\ell=2$ and $k=10$. Since a diamond has a radius of at most $2\stl+1$ and $\ell\leq k\stl$, $k(4\stl+1)+(2k\stl+3)(4k\stl+5)$ is an upper bound on $n$.

    The reason why we place copies of $R$ on the board $B$ is to make sure that at least $2\ell+2$ moves are required to flood $B$, even in the absence of diamonds. To see this, suppose first that we flood fill from the top left square in each move. From the definition of the arches of $R$, disregarding the diamonds on $B$, a minimal sequence of moves will consist of $\ell$ 1s or 2s interspersed with a total of $\ell-1$ 3s, followed by the three moves 3, 2 and 1, respectively. Note that only one copy of $R$ on $B$ would be enough to achieve this. However, having several copies of $R$ on $B$ does not affect the minimum number of moves as all copies will get flooded simultaneously. The idea with the $2\ell+3$ copies of $R$ is to make sure that at least $2\ell+2$ moves are required to flood $B$ even when we are allowed to choose which tile to flood fill from in each move. To see this, suppose that we choose to flood fill from a tile inside one of the copies of $R$. Since there are $2\ell+3$ copies, similar reasoning to the end of the proof of Lemma~\ref{lem:NPC-four} tells us that we will do worse than $2\ell+2$ moves.


    We will now argue that the number of moves required to flood $B$ is $2\ell+2$ if and only if there is a common supersequence of $s_1,\dots,s_k$ of length at most $\ell$. We choose to flood fill from the top left tile in each move.

    Suppose first that there is a common supersequence $s$ of length $\ell'\leq \ell$. Let $s'$ be the string $s$ followed by $\ell-\ell'$~1s. Let $s''$ be the $(2\ell+2)$-long string obtained from $s'$ by inserting a~3 after each character of $s'$ and adding the two additional characters 2 and 1 to the end. We make $2\ell+2$ moves by choosing the colours in the same order as they appear in $s''$. Note that all diamonds are flooded after $2\ell'$ moves, and by the last move we have also flooded every copy of $R$, and hence the whole board $B$.

\begin{figure}[t]
\centering
	\tikzstyle{box}=[thick]
	\begin{tikzpicture}[inner sep=2pt,scale=0.255]

	\scope[shift={(45,-2)}]

		\draw[] (0,0) rectangle (14,-14);

		\fill[color=gray!50] (0,-14) rectangle +(2,1);
		\fill[color=gray!50] (2,-14) rectangle +(2,1);
		\fill[color=gray!50] (4,-14) rectangle +(2,1);
		\fill[color=gray!50] (6,-14) rectangle +(2,1);
		\fill[color=gray!50] (8,-14) rectangle +(2,1);
		\fill[color=gray!50] (10,-14) rectangle +(2,1);
		\fill[color=gray!50] (12,-14) rectangle +(2,1);

		\draw[] (0,-14) rectangle +(2,1);
		\draw[] (2,-14) rectangle +(2,1);
		\draw[] (4,-14) rectangle +(2,1);
		\draw[] (6,-14) rectangle +(2,1);
		\draw[](8,-14) rectangle +(2,1);
		\draw[](10,-14) rectangle +(2,1);
		\draw[] (12,-14) rectangle +(2,1);

		\fill[color=gray!50] (4,-3) -- ++(-1,-1) -- ++(-1,1) -- ++(1,1) -- ++(1,-1);
		\fill[color=gray!50] (4,-7) -- ++(-1,-1) -- ++(-1,1) -- ++(1,1) -- ++(1,-1);
		\fill[color=gray!50] (6,-5) -- ++(-1,-1) -- ++(-1,1) -- ++(1,1) -- ++(1,-1);
		\fill[color=gray!50] (8,-7) -- ++(-1,-1) -- ++(-1,1) -- ++(1,1) -- ++(1,-1);
		\fill[color=gray!50] (10,-9) -- ++(-1,-1) -- ++(-1,1) -- ++(1,1) -- ++(1,-1);
		\fill[color=gray!50] (8,-11) -- ++(-1,-1) -- ++(-1,1) -- ++(1,1) -- ++(1,-1);
		\fill[color=gray!50] (8,-3) -- ++(-1,-1) -- ++(-1,1) -- ++(1,1) -- ++(1,-1);
		\fill[color=gray!50] (10,-5) -- ++(-1,-1) -- ++(-1,1) -- ++(1,1) -- ++(1,-1);
		\fill[color=gray!50] (12,-7) -- ++(-1,-1) -- ++(-1,1) -- ++(1,1) -- ++(1,-1);
		\fill[color=gray!50] (12,-3) -- ++(-1,-1) -- ++(-1,1) -- ++(1,1) -- ++(1,-1);

		\draw[] (4,-3) -- ++(-1,-1) -- ++(-1,1) -- ++(1,1) -- ++(1,-1);
		\draw[] (4,-7) -- ++(-1,-1) -- ++(-1,1) -- ++(1,1) -- ++(1,-1);
		\draw[] (6,-5) -- ++(-1,-1) -- ++(-1,1) -- ++(1,1) -- ++(1,-1);
		\draw[] (8,-7) -- ++(-1,-1) -- ++(-1,1) -- ++(1,1) -- ++(1,-1);
		\draw[] (10,-9) -- ++(-1,-1) -- ++(-1,1) -- ++(1,1) -- ++(1,-1);
		\draw[] (8,-11) -- ++(-1,-1) -- ++(-1,1) -- ++(1,1) -- ++(1,-1);
		\draw[] (8,-3) -- ++(-1,-1) -- ++(-1,1) -- ++(1,1) -- ++(1,-1);
		\draw[] (10,-5) -- ++(-1,-1) -- ++(-1,1) -- ++(1,1) -- ++(1,-1);
		\draw[] (12,-7) -- ++(-1,-1) -- ++(-1,1) -- ++(1,1) -- ++(1,-1);
		\draw[] (12,-3) -- ++(-1,-1) -- ++(-1,1) -- ++(1,1) -- ++(1,-1);
	\endscope

	\end{tikzpicture}
\caption{A board constructed in the proof of Lemma~\ref{lem:NPC-three}.}
\label{fig:NPC2}
\end{figure}
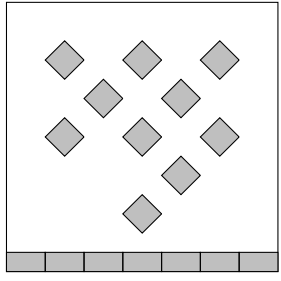

    Suppose second that $B$ can be flooded in $2\ell+2$ moves. The centre of each diamond has the colour~3 and therefore the first $2\ell$ moves flood the diamonds. The subsequence of these first $2\ell$ moves induced by the colours 1 and 2 is an $\ell$-long common supersequence of $s_1,\dots,s_k$.
    \qed
\end{proof}

We can now summarise Lemmas~\ref{lem:NPC-four} and~\ref{lem:NPC-three} in the following theorem.

\begin{theorem}
    \label{thm:NPC}
    For $c\geq 3$, \Coloroid{c} and \ColoroidFree{c} are \NPtime-hard (and the decision versions are \NPtime-complete). Further, for an unbounded number of colours $c$, there is no polynomial-time constant factor approximation algorithm, unless $\Ptime=\NPtime$.
\end{theorem}

For two colours, \Coloroid{2} is trivially in \Ptime, but it is not that obvious what the complexity of \ColoroidFree{2} is. The next theorem settles this question, by showing that an optimal strategy for any instance of \ColoroidFree{2} consists of flooding from the same tile in each move.

\begin{theorem}\label{thm:NPC-freetwo}
        \ColoroidFree{2} is in \Ptime.
\end{theorem}
\begin{proof}
    We first consider the case where we are allowed to flood fill from exactly two distinct tiles of the board. At the end of the proof we consider the case where flooding from any tile is allowed.

    Suppose there exists a shortest sequence of moves $S$ that floods the board from only two tiles $t_1$ and $t_2$. Suppose also that $t_1$ and $t_2$ belong to different connected components during the first $m_1+m_2$ moves but become connected in the $(m_1+m_2+1)$th move, where $m_1$ is the number of flood filling operations from $t_1$ and $m_2$ is the number of flood filling operations from $t_2$. Suppose without loss of generality that in move $m_1+m_2+1$, we flood fill from $t_1$. Let $t'_1$ and $t'_2$ be two adjacent tiles such that after $m_1+m_2$ moves, $t_1$ and $t'_1$ belong to the same monochromatic region, and $t_2$ and $t'_2$ belong to the same monochromatic region. Let $P_1$ be a simple path from $t_1$ to $t'_1$ in the board with the monochromatic connected components $\alpha_0,\dots,\alpha_{m_1}$, such that the $i$th flood filling move from $t_1$ merges $\alpha_i$ with the monochromatic region that contains $t_1$. Thus, $t_1\in \alpha_0$ and the whole path $P_1$ is monochromatic after $m_1$ flood filling operations from $t_1$. We define a path $P_2$ from $t_2$ to $t'_2$ similarly. Let $\beta_0,\dots,\beta_{m_2}$ be the monochromatic connected components of $P_2$. Figure~\ref{fig:NPC-2freeA} illustrates the two paths $P_1$ and $P_2$.
    \begin{figure}[t]
        \centering
	\tikzstyle{box}=[very thick]	
\begin{tikzpicture}[inner sep=2pt,scale=0.255]

	\scope[shift={(0,0)}]

		\draw (9,2) node {(a)};
		\drawrow{1} {222222222222222222}
		\drawrow{2} {222332222222332222}
		\drawrow{3} {233223322222333222}
		\drawrow{4} {322332233222222222}
		\drawrow{5} {323223323322222222}
		\drawrow{6} {233232322322222222}
		\drawrow{7} {232323232322332222}
		\drawrow{8} {233232323223333322}
		\drawrow{9} {323223232223223322}
		\drawrow{10}{232332322332322332}
		\drawrow{11}{232223233223233232}
		\drawrow{12}{223333322232323232}
		\drawrow{13}{222222233223232332}
		\drawrow{14}{222222223322332332}
		\drawrow{15}{223322222332222322}
		\drawrow{16}{233332222333333322}
		\drawrow{17}{223222222233333222}
		\drawrow{18}{222222222222222222}

		\drawbox{18}{18};
		\draw[thick,solid] (5,-6.5) -- ++(0.5,0) -- ++(0,-5) --++ (0.5,0);
		\draw[thick,dashed] (8,-11.5) -- ++(4,0);
		\draw[box] (4,-6) rectangle +(1,-1);
		\draw[box] (6,-11) rectangle +(1,-1);

		\draw[box] (7,-11) rectangle +(1,-1);
		\draw[box] (12,-11) rectangle +(1,-1);
	\endscope

\scope[shift={(22,0)}]

		\draw (9,2) node {(b)};
		\drawrow{1} {222222222222222222}
		\drawrow{2} {222332222222332222}
		\drawrow{3} {233333322222333222}
		\drawrow{4} {333333333222222222}
		\drawrow{5} {333333333322222222}
		\drawrow{6} {233333333322222222}
		\drawrow{7} {233333333322332222}
		\drawrow{8} {233333333223333322}
		\drawrow{9} {333333332223223322}
		\drawrow{10}{233333322332222332}
		\drawrow{11}{233333233222222232}
		\drawrow{12}{223333322222222232}
		\drawrow{13}{222222233222222332}
		\drawrow{14}{222222223322222332}
		\drawrow{15}{223322222332222322}
		\drawrow{16}{233332222333333322}
		\drawrow{17}{223222222233333222}
		\drawrow{18}{222222222222222222}

		\drawbox{18}{18};
		\draw[thick,solid] (5,-6.5) -- ++(0.5,0) -- ++(0,-5) --++ (0.5,0);
		\draw[thick,dashed] (8,-11.5) -- ++(4,0);
		\draw[box] (4,-6) rectangle +(1,-1);
		\draw[box] (6,-11) rectangle +(1,-1);

		\draw[box] (7,-11) rectangle +(1,-1);
		\draw[box] (12,-11) rectangle +(1,-1);
	\endscope
	\end{tikzpicture} 
        \caption{ A board (a) before and (b) after $m_1+m_2$ moves as discussed in the proof of Theorem~\ref{thm:NPC-freetwo}. The solid and dashed paths give $P_1$ and $P_2$ respectively. In left-to-right order, the emphasised tiles are $t_2,t_2',t_1'$ and $t_1$.}
    \label{fig:NPC-2freeA}
    \end{figure}
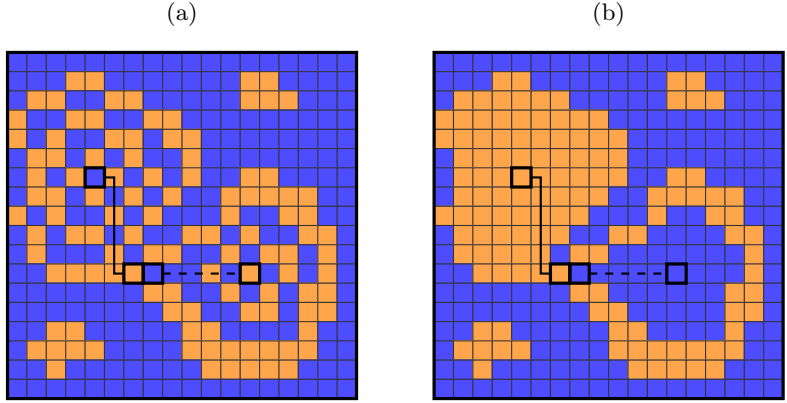

    We now show that the area flooded after the first $m_1+m_2+1$ moves of $S$ can be flooded with $m_1+m_2+1$ flood filling moves from one single tile $t_3$. Let $P_3$ be the path $P_1$ concatenated with a reversed copy of $P_2$. Thus, the monochromatic connected components $\gamma_i$ of $P_3$ are $\gamma_0=\alpha_0,\, \gamma_1=\alpha_1,\dots, \gamma_{m_1}=\alpha_{m_1},\, \gamma_{m_1+1}=\beta_{m_2},\, \gamma_{m_1+2}=\beta_{m_2-1}, \dots, \gamma_{m_1+m_2+1}=\beta_0$. Let $t_3$ be a tile in $\gamma_{m_2}$ and consider a series of flood filling moves from this tile: after the first $m_2$ moves, $t_1$ and $t_3$ are connected, and after the first $m_1+1$ moves, $t_2$ and $t_3$ are connected. Once a tile $t$ is in the same monochromatic component as $t_3$, flooding from $t_3$ is equivalent to flooding from $t$. Thus, after a total of $m_1+m_2+1$ flood filling moves from $t_3$, we have effectively performed $m_1+1$ flood filling moves from $t_1$ and $m_2$ flood filling moves from $t_2$. This is exactly what the first $m_1+m_2+1$ moves of $S$ do. Hence we can replace the moves in $S$ by flooding from a single tile $t_3$.

    Finally we deal with the case where we are allowed to flood fill from any tile. Consider a shortest sequence $S$ of moves that flood the board and suppose that we flood fill from the tiles $t_1,\dots,t_r$, for $r>2$. Suppose without loss of generality that the first merge of any of these tiles is when we flood fill from $t_1$, which connects $t_1$ with $t_2,\dots,t_{r'}$, where $2\leq r'\leq r$. Let $m_i$ be the number of flood filling operations that have taken place from $t_i$ before this merge. The following sequence $S'$ of moves will flood the board in at most $|S|$ moves but flood fills from only $r-1$ tiles. For $i=3,\dots,r$, first perform $m_i$ flood filling operations from $t_i$. Instead of flooding from $t_1$ and $t_2$ separately, we use the result above and flood fill from a different tile $t'$. Thus, by the next $m_1+m_2+1$ moves we have connected $t_1,\dots,t_{r'}$. The subsequent moves of $S'$ follow those of $S$, where any move at $t_1$ or $t_2$ is replaced with a move at $t'$. Inductively we reduce the number of tiles to flood fill from to a single tile. The conclusion is that we can solve \ColoroidFree{2} by attempting to flood the entire board from each tile of the board in turn, which requires only polynomial time.
    \qed
\end{proof}

\section{The complexity of constant height boards} \label{sec:shape}

So far we have analysed the complexity of $\Floodit$ on square shaped \recdim{n}{n} boards. A natural question to ask is: what is the complexity of \Coloroid{c} on an \recdim{h}{n} board, where the height $h$ is a fixed constant? We denote this problem by \Coloroid{(c,h)} and the `free' variant by \ColoroidFree{(c,h)}, analogously.

\Coloroid{(c,1)} is trivially in \Ptime, and Fleischer and Woeginger have shown (personal communication) that \ColoroidFree{(c,1)} is also in \Ptime. We will show that of \Coloroid{(c,2)} on a \recdim{2}{n} board remains in \Ptime. However, the complexity of \ColoroidFree{(c,2)} remains unresolved. Before stating this result we will prove in Theorem~\ref{thm:hNPC} that when the number of colours is unbounded and $h\geq 3$ then both \Coloroid{(c,h)} and \ColoroidFree{(c,h)} are \NPtime-hard.


For the \Coloroid{c} problem on a square \recdim{n}{n} board with $c \geq 4$ we gave a reduction from the shortest common supersequence problem (\SCS{}) which embedded a number of diamond structures into a board filled with a single background colour. Each diamond represented one of the strings in the \SCS{} instance. The problem with this reduction on an \recdim{h}{n} board is that a string of length $\ell$ was represented by a diamond with height $4\ell-1$. This is not possible if $h<4\ell-1$. However, Timkovskii proved~\cite{Timkovskii:1989} that the \SCS{} problem remains \NPtime-hard even when the length of the strings is constrained to be at most $2$, and the alphabet size is unbounded. Inspection of the proof of Lemma~\ref{lem:NPC-four} shows that \Coloroid{(c,h)} is \NPtime-hard (and the decision version \NPtime-complete) when $h \geq 8$ and the number of colours is unbounded. Naively, it would appear that $h=7$ suffices in the proof of Lemma~\ref{lem:NPC-four} as it allows enough height to embed a diamond representing a string of length $2$ as is required. However, for the reduction to be valid we also need to leave at least one row of space above the diamonds so that all diamonds can be flooded simultaneously on any move.

To reduce the board height required for our \NPtime-hardness proof further we reduce the height of the diamond structures used in the reduction. Recall that the reduction begins by doubling the length of all strings in a way that ensures that no string contains a character which is followed immediately by another occurrence of the same character. We now show in Lemma~\ref{lem:SCSab} that the \SCS{} problem remains \NPtime-hard even when all the strings are of the form $ab$ where $a,b \in \Sigma$ and $a \neq b$. The proof is by reduction from the \SCS{} problem with the constraint that all strings have length at most $2$. This result allows us to remove the doubling step and reduce the height of the diamond structures, resulting in Theorem~\ref{thm:hNPC} which gives the desired result.

\begin{lemma}
    \label{lem:SCSab}
    The \SCS{} problem is \NPtime-hard when all the strings are of the form $ab$ where $a,b \in \Sigma$ and $a \neq b$.
\end{lemma}
\begin{proof}
Let $S$ be an instance of \SCS{} that contains $k$ strings $s_1,\ldots,s_k$ each of length $w\leq 2$ over the alphabet $\Sigma$. We abuse notation by referring to $S$ as both the instance and the set of $k$ strings. We begin by assuming that $S$ contains a string of length $1$. Without loss of generality let $s_k$ be such a string. Let $S'$ be the instance of \SCS{} formed by the $k-1$ strings $s_1,\ldots s_{k-1}$. There are two cases to consider. In the first case, the single character, $a$, in $s_k$ occurs in some string $s_1,\ldots,s_{k-1}$. Therefore any common supersequence of $S'$ contains an $a$ and hence is also a common supersequence of $S$. Further, as $S$ is a superset of $S'$, any common supersequence of $S$ is a common supersequence of $S'$. Hence $|\SCS{}(S)|=|\SCS{}(S')|$. In the second case, the single character, $a$, in $s_k$ does not occur in $s_1,\ldots,s_{k-1}$. A common supersequence for $S$ can therefore be found by inserting $a$ at the end of the shortest common supersequence for $S'$. Hence $|\SCS{}(S)|\leq |\SCS{}(S')|+1$. Further, any common supersequence for $S$ must contain an $a$ and is also a common supersequence for $S'$. Therefore by removing the $a$ we have that $|\SCS{}(S)|\geq |\SCS{}(S')|+1$. Hence in this case $|\SCS{}(S)| = |\SCS{}(S')|+1$. Repeated application of the above technique gives a poly-time reduction from the \SCS{} problem with strings of length $w \leq 2$ to  the \SCS{} problem with strings of length $w = 2$. Therefore we have that the latter is also \NPtime-hard.

We now redefine $S$ to be an instance of \SCS{} that contains $k$ strings $s_1,\ldots,s_k$ each of length exactly $2$ over the alphabet $\Sigma$. We begin by assuming that $S$ contains a string of the form $aa$ where $a \in \Sigma$. Without loss of generality let $s_k$ be such a string. Let $S'$ be the instance of \SCS{} formed by the $k-1$ strings $s_1,\ldots s_{k-1}$ and new strings $s_{k+1}=aa'$ and $s_{k+2}=a'a$ where $a'$ does not occur in $S$. First consider the shortest common supersequence of $S$, which must contain $aa$ as a subsequence. By inserting $a'$ between these two occurrences of $a$, we obtain a common supersequence of $S'$ of length $|\SCS{}(S)|+1$. Therefore $|\SCS{}(S')| \leq |\SCS{}(S)|+1$. Now consider the shortest common supersequence of $S'$, which must contain either $aa'a$ or $a'aa'$ as a subsequence. In the former case by removing the $a'$ symbol we obtain a common supersequence of $S$ and have that $|\SCS{}(S')| \geq |\SCS{}(S)|+1$. In the latter, when we remove the two occurrences of $a'$ we obtain a common supersequence of $s_1,\ldots, s_{k-1}$ of length $|\SCS{}(S')|-2$. This sequence contains exactly one occurrence of $a$, and by inserting a second we obtain a common supersequence of $S$ of length $|\SCS{}(S')|-1$. Therefore $|\SCS{}(S')| = |\SCS{}(S)|+1$. Repeated application of the above technique gives a poly-time reduction from the \SCS{} problem with strings of length $2$ to  the \SCS{} problem with strings of the form $ab$ where $a,b \in \Sigma$ and $a \neq b$. Therefore we have that the latter is also \NPtime-hard. \qed
\end{proof}

\begin{theorem}
    \label{thm:hNPC}
     \Coloroid{(c,h)} and \ColoroidFree{(c,h)} are \NPtime-hard when $h\geq 3$ and the number of colours $c$ is unbounded (and the decision versions are \NPtime-complete).
\end{theorem}
\begin{proof}
First observe that the decision versions of both problems are in \NPtime{} because the unconstrained versions, \Coloroid{c} and \ColoroidFree{c}, are in \NPtime{}. We begin by considering the \Coloroid{(c,h)} problem for $h \geq 3$. We reduce from an instance of \SCS{} on $k$ strings $s_1,\dots,s_k$ over the alphabet $\Sigma$ and an integer $\ell$. The strings are constrained to have the form $ab$ where $a,b \in \Sigma$ and $a \neq b$. Let $B$ be a \recdim{h}{n} board filled with a single background colour where $n=4k+1$. For each symbol in $\Sigma$ we have a corresponding distinct colour in addition to the background colour. For each string $s_i = a_ib_i$ we embed a `half' diamond against the bottom edge of the board. The half diamond consists of a single tile of colour $b_i$ (the inner layer), surrounded on all three sides by a tile of colour $a_i$ (the outer layer). This is illustrated in Figure~\ref{fig:height3} for $h=3$.

\begin{figure}[t]
    \centering
	\tikzstyle{box}=[very thick]	
\begin{tikzpicture}[inner sep=2pt,scale=0.4]

	\scope[shift={(0,0)}]

		\drawrow{1} {88888888888888888}
		\drawrow{2} {88288818883888288}
		\drawrow{3} {82328121832382128}

		\drawrowL{1} {00000000000000000}
		\drawrowL{2} {00200010003000200}
		\drawrowL{3} {02320121032302120}
		\drawbox{17}{3}
	\endscope

	\end{tikzpicture} 
    \caption{An example of a board constructed in the proof of Theorem~\ref{thm:hNPC}. In left-to-right order, the strings embedded are ``23'',``12'',``32'' and ``21''. The shortest common supersequence is ``2132''.}
    \label{fig:height3}
\end{figure}

Observe that as $h \geq 3$ and $n=4k+1$, all the half diamonds can be placed so that the outer layer of each half diamond is surrounded by the background colour. Therefore on any move, the outer layer of any half diamond can be flooded. Further observe that for all $i$, as $a_i \neq b_i$ the diamond for $s_i$ is flooded if and only if the move sequence contains $a_ib_i$ as a subsequence. Therefore a move sequence floods the board if and only if it is a common supersequence of $s_1,\dots,s_k$, so $|\SCS{(S)}|$ equals the length of the shortest move sequence which floods the board. As this reduction can be implemented in polynomial time, we have that \Coloroid{(c,h)} problem is \NPtime-hard with an unbounded number of colours.

We now consider the \ColoroidFree{(c,h)} problem for $h \geq 3$. \NPtime-hardness follows by the same argument as for the \NPtime-hardness of \ColoroidFree{c} for $c \geq 4$ given in the proof of Lemma~\ref{lem:NPC-four}. We increase the size of the board (horizontally) and embed $2k+1$ copies of each half diamond. We observe that any flood-filling move begun from a tile in an unflooded half diamond floods only tiles in that half diamond. This ensures that any move sequence which floods the board and contains moves begun from tiles in an unflooded half diamond either contains at least $2k+1$ moves or contains redundant moves. In either case, it is not minimal. \qed
\end{proof}


We finally show that \Coloroid{(c,2)} is in \Ptime.

\begin{theorem}
    \label{thm:height2}
    For any $c\geq 1$, \Coloroid{c} on a \recdim{2}{n} board is in \Ptime. More precisely, the running time is $O(n)$.
\end{theorem}
\begin{proof}
    Suppose that $B$ is a \recdim{2}{n} board and $c$ is the number of colours. We say that a tile $t$ on $B$ is \emph{marked} if it has colour $c_t$ and no other tile in the columns strictly to the right of $t$ has the colour $c_t$. A column is marked if it contains a marked tile.

    The key observation, which holds on a \recdim{2}{n} board, is that if the marked tiles are flooded then so is the whole board $B$. To see this, note that when a marked tile $t$ of colour $c_t$ is flooded, all other tiles of the colour $c_t$ that have not yet been flooded are to the left of $t$ and therefore adjacent to the flooded region. Hence they will be flooded when $t$ is flooded. Thus, we ask for the shortest sequence of moves that would flood the marked tiles.

    A \emph{shortest path} to a tile $t$ denotes a shortest sequence of flood filling operations that includes  $t$ in the flooded region. If $t$ is already included in the flooded region, then the length of the shortest path to $t$ is~0.

    One might think that a solution to \Coloroid{c} on a \recdim{2}{n} board would be to go from one marked tile to the next in left-to-right order using shortest paths. Although this is correct, we must be a little careful with which shortest paths we choose. The following procedure floods the marked tiles in the smallest number of moves possible.

    \begin{description}
        \item{\bf Beginning of procedure.} Let $i$ be the leftmost marked column such that $i$ contains a marked tile $t$ that has not yet been flooded. Let $t'$ be the other tile in column $i$. We have two cases.
            \begin{description}
                \item{\bf Case 1 ($t'$ is unmarked).} Let $m$ and $m'$ be the lengths of the shortest paths to $t$ and $t'$, respectively. Note that $|m-m'|\leq 1$. We consider two subcases.
                    \begin{description}
                        \item{\bf Case 1a ($m\leq m'$).} Flood using the sequence of colours found along the shortest path to $t$, then go to the beginning of the procedure. Correctness: Flooding $t'$ before $t$ means that we are bound to flood $t$ at a later stage. Once $t'$ is flooded we can never do worse by flooding $t$ immediately. Thus, flooding $t'$ before $t$ and then flooding $t$ takes a total of at least $m+1$ moves. However, flooding $t$ takes $m$ moves and we are not necessarily forced to spend an extra move on flooding $t'$, which is not a marked tile.
                        \item{\bf Case 1b ($m>m'$).}  Flood using the sequence of colours found along the shortest path to $t'$ and then flood $t$. Then go to the beginning of the process. Correctness: Flooding $t$ takes at at least $m'+1$ steps, even if we do not go via $t'$. Since all remaining marked tiles are to the right of column $i$, we should therefore flood $t'$ before $t$. Once $t'$ is flooded, we can never do worse by flooding $t$ immediately.
                \end{description}
                \item{\bf Case 2 ($t'$ is marked).} Flood using the sequence of colours found along the shortest of the shortest paths to $t$ or $t'$. Then flood the remaining tile in column $i$. Then go to the beginning of the process. Correctness: Both $t$ and $t'$ must eventually be flooded. Once one of them is flooded, there is no reason to wait to flood the other.
            \end{description}
    \end{description}

    Using for example dynamic programming, the shortest path to a tile $t$ on a \recdim{2}{n} board can be computed in time linear in the distance between the flooded region and $t$. We note that the shortest paths are always calculated between the rightmost end of the flooded region and a marked column~$i$. Since the flooded region is always extended to column $i$ in each step of the procedure, the total running time of computing the shortest paths is linear in $n$. Hence the running time of the whole process is $O(n)$. \qed
\end{proof}

\section{Approximating the number of moves} \label{sec:approx}
As we have seen, \Coloroid{c} and \ColoroidFree{c} are not efficiently approximable to within a constant factor for an unbounded number of colours $c$. However, a $(c-1)$-approximation for \Coloroid{c}, $c \geq 3$, can easily be obtained as follows. Suppose that $B$ is a board on the colours $1,\dots,c$. Clearly, if we repeatedly cycle through the sequence of colours $1,\dots,c$  then $B$ will be flooded after at most $c\times m(B)$ moves. We can do a little better by first cycling through the ordered sequence of colours $1,\dots,c$ and then repeatedly alternating between a cycle of the sequence $(c-1),\dots,1$ and a cycle of $2,\dots,c$ until there are only two distinct colours left on the board, after which we alternate between the two remaining colours. Note that there are always exactly two distinct colours left before the final move. The board $B$ is guaranteed to be flooded after at most $c+(c-1)(m(B)-2)+1\leq (c-1)m(B)$ moves, which gives us a $(c-1)$-approximation algorithm.

A randomised approach with an expected number of moves of approximately $2c/3\times m(B)$ is obtained as follows. Suppose that $s$ is a minimal sequence of colours that floods $B$ (flood filling from the top left square in each move). We shuffle the $c$ colours and process them one by one. If $B$ is not flooded then we shuffle again and repeat. Note that this procedure could (and most likely will) generate many useless moves that do not merge any monochromatic regions. Thus, if $m(B)=1$ then the algorithm could take up to $c$ moves, although a single move would suffice. If $m(B)=2$ then $c + \frac{1}{2} c=3c/2$ is an upper bound on the expected number of moves; with probability $1/2$, the two moves in $s$ appear in the same order as in the shuffled sequence of colours, and if not, we might have to shuffle the colours again and repeat one last time. We generalise this as follows. Let $T(m)$ be (an upper bound on) the expected number of moves it takes to produce a fixed sequence of $m$ moves. We have $T(m)=c+\frac{1}{2}T(m-1)+\frac{1}{2}T(m-2)$. Solving the recurrence with the values of $T(1)$ and $T(2)$ above gives us a solution in which $T(m)$ is asymptotically $(2c/3)m$ for a fixed $c$.
%
%
%

\section{General bounds on the number of moves} \label{sec:bound}
Recall that we denote the minimum number of moves which flood some board $B$ as $m(B)$. In this section we investigate bounds on the maximum $m(B)$ over all boards in $B_{n,c}$ which we denote \Mnc{}. Intuitively, this can be seen as the minimum number of moves to flood the `worst' board in $B_{n,c}$.

For motivation, consider an \recdim{n}{n} checker board of two colours as shown in Figure~\ref{fig:uplow}. First observe that as the board has only two colours, the player has no choice in their next move. Consider a diagonal of tiles in the direction top-right to bottom-left where the 0th diagonal is the top-left corner. Further observe that move $k$ floods exactly the $k$th diagonal, so the total number of moves is $2(n-1)$. Thus we have shown that $\Mnc \geq 2(n-1)$.

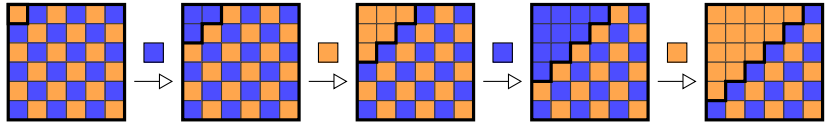
\begin{figure}[t]
\centering
	\tikzstyle{box}=[very thick]
	\begin{tikzpicture}[inner sep=2pt,scale=0.255]

	\scope[shift={(0,0)}]

  		\drawrow{1}{323232};
		\drawrow{2}{232323};
		\drawrow{3}{323232};
		\drawrow{4}{232323};
		\drawrow{5}{323232};
		\drawrow{6}{232323};

		\draw[box] (0,0) rectangle +(6,-6);
		\draw[box] (0,0) rectangle +(1,-1);
	\endscope

	\fill[c2] (7,-3) rectangle +(1,1);
	\draw[] (7,-3) rectangle +(1,1);

	\draw[-open triangle 60] (6.5,-4) -- +(2,0);
	\scope[shift={(9,0)}]
  		\drawrow{1}{223232};
		\drawrow{2}{232323};
		\drawrow{3}{323232};
		\drawrow{4}{232323};
		\drawrow{5}{323232};
		\drawrow{6}{232323};

		\draw[box] (0,0) rectangle +(6,-6);
		\draw[box] (2,0) -- ++(0,-1) -- ++(-1,0) -- ++(0,-1) -- ++(-1,0);
	\endscope

	\fill[c3] (16,-3) rectangle +(1,1);
	\draw[] (16,-3) rectangle +(1,1);

	\draw[-open triangle 60] (15.5,-4) -- +(2,0);

	\scope[shift={(18,0)}]
  		\drawrow{1}{333232};
		\drawrow{2}{332323};
		\drawrow{3}{323232};
		\drawrow{4}{232323};
		\drawrow{5}{323232};
		\drawrow{6}{232323};

		\draw[box] (0,0) rectangle +(6,-6);
		\draw[box] (3,0) -- ++(0,-1) -- ++(-1,0) -- ++(0,-1) -- ++(-1,0) -- ++(0,-1) -- ++(-1,0); (-2,0);
	\endscope

	\fill[c2] (25,-3) rectangle +(1,1);
	\draw[] (25,-3) rectangle +(1,1);

	\draw[-open triangle 60] (24.5,-4) -- +(2,0);

	\scope[shift={(27,0)}]
  		\drawrow{1}{222232};
		\drawrow{2}{222323};
		\drawrow{3}{223232};
		\drawrow{4}{232323};
		\drawrow{5}{323232};
		\drawrow{6}{232323};

		\draw[box] (0,0) rectangle +(6,-6);
		\draw[box] (4,0) -- ++(0,-1) -- ++(-1,0) -- ++(0,-1) -- ++(-1,0) -- ++(0,-1) -- ++(-1,0) -- ++(0,-1) -- ++(-1,0);
	\endscope

	\fill[c3] (34,-3) rectangle +(1,1);
	\draw[] (34,-3) rectangle +(1,1);

	\draw[-open triangle 60] (33.5,-4) -- +(2,0);

	\scope[shift={(36,0)}]
  		\drawrow{1}{333332};
		\drawrow{2}{333323};
		\drawrow{3}{333232};
		\drawrow{4}{332323};
		\drawrow{5}{323232};
		\drawrow{6}{232323};

		\draw[box] (0,0) rectangle +(6,-6);
		\draw[box] (5,0) -- ++(0,-1) -- ++(-1,0) -- ++(0,-1) -- ++(-1,0) -- ++(0,-1) -- ++(-1,0) -- ++(0,-1) -- ++(-1,0) -- ++(0,-1) -- ++(-1,0);
	\endscope

	\end{tikzpicture} 
\caption{Progression of a \recdim{6}{6} checker board. }
\label{fig:uplow}
\end{figure}

We now give an overview of a simple algorithm which floods any board in $B_{n,c}$ in at most $c(n-1)$ moves. The algorithm performs $n$ stages. The purpose of the $i$th stage is to flood the $i$th row.  Stage $i$ repeatedly picks the colour of the leftmost tile in row $i$ which is not in the flooded region, until row $i$ is flooded.

First observe that Stage $1$ performs at most $n-1$ moves to flood row $i$ (we can flood at least one tile of row $1$ per move). When the algorithm begins Stage $i\geq 2$, observe that row $i-1$ is entirely flooded as well as any tiles in row $i$ which match the colour of row $i-1$. Therefore when a new colour is selected, all tiles in row $i$ of this colour become flooded. Hence at most $c-1$ moves are performed by Stage $i$. Summing over all rows, this gives the desired bound that $\Mnc \leq c(n-1)$. Observe that from the previous example with the checker board on $c=2$ colours, the bound $c(n-1)$ is tight. Thus, the checker board is the `worst' board in $B_{n,2}$.

As motivation, we have given weak bounds on \Mnc{}. We now tighten these bounds for large $c$ by providing a better algorithm for flooding an arbitrary board. We will also give a description of `bad' boards which require many moves to be flooded. It will turn out that \Mnc{} is asymptotically $\Theta(\sqrt{c}\,n)$ for increasing $n$ and~$c$.

\newcommand{\rem}{\text{rem}}
\begin{theorem} \label{thm:good}
There exists a polynomial time algorithm for \Floodit{} which can flood any \recdim{n}{n} board with $c$ colours in at most $2n+(\sqrt{2c})n+c$ moves.
\end{theorem}

\begin{proof}

For a given integer $\ell$ (to be determined later), we partition the board horizontally into $\ell+1$ contiguous sections, denoted $S_0,\dots,S_\ell$ from top to bottom, as follows. Let $q=\lfloor n/\ell \rfloor$ and $r=n \mod \ell$. Section $S_0$ consists of the first $\lceil q/2 \rceil$ rows, $S_1,\dots,S_r$ contain $(q+1)$ rows each (if $r>0$), and $S_{r+1},\dots,S_{\ell-1}$ contain $q$ rows each (if $r<\ell-1$). Section $S_\ell$ contains $\lfloor q/2 \rfloor$ rows. See Figure~\ref{fig:upper} for an illustration. We let $y(i)$ denote the final row of $S_i$.

The algorithm performs the following three stages.
\begin{description}
    \item{Stage~1.} Flood the first column.
    \item{Stage~2.} Flood row $y(x)$ for all $0\leq x < \ell$.
    \item{Stage~3.} Cycle through the $c$ colours until the board is flooded.
\end{description}

The correctness of our algorithm is immediate as Stage~$3$ ensures that the board is flooded by cycling colours. Stage~$1$ can be implemented to perform at most $n-1$ moves as argued for the simple algorithm above. Similarly, Stage~$2$ can be completed in $\ell (n-1)$ moves. We now analyse Stage~$3$.
\begin{figure}[t]
    \begin{minipage}[b]{0.50\linewidth}
        \centering
        \tikzstyle{box}=[thick]
	\begin{tikzpicture}[inner sep=2pt,scale=0.16]

	\scope[shift={(0,0)}]

		\draw[box] (0,0) rectangle (25,25);
		\draw[box] (0,0) rectangle (1,25);

		\draw[box] (1,0) rectangle (25,4); 

		\draw[box] (1,4) rectangle (25,5);
		\draw[box] (1,5) rectangle (25,12); 

		\draw[box] (1,12) rectangle (25,13);
		\draw[box] (1,13) rectangle (25,21); 

		\draw[box] (1,21) rectangle (25,22); 
		\draw[box] (1,22) rectangle (25,25);

 		\draw (12,23.5) node {$S_0$};
 		\draw (12,17) node {$S_1$};
 		\draw (12,8.5) node {$S_2$};
 		\draw (12,2) node {$S_3$};

		\draw (-2,21.5) node {$y(0)$};
		\draw (-2,12.5) node {$y(1)$};
		\draw (-2,4.5)  node {$y(2)$};

		\draw[|-|] (26,25) -- (26,21);
		\draw[|-|] (26,21) -- (26,12);
		\draw[|-|] (26,12) -- (26,4);
		\draw[|-|] (26,4) -- (26,0);

 		\draw (29,23) node {$\lceil q/2 \rceil$};
 		\draw (29,17) node {$q+1$};
 		\draw (29,8.5) node {$q$};
 		\draw (29,2) node {$\lfloor q/2 \rfloor$};
	\endscope

	\end{tikzpicture} 
        \caption{The board decomposition used in the proof of Theorem~\ref{thm:good}.\label{fig:upper}}
    \end{minipage}
    \hspace{0.06\linewidth}
    \begin{minipage}[b]{0.44\linewidth}
        \centering
        	\begin{tikzpicture}[inner sep=2pt,scale=0.2]

%
%

	\scope[shift={(0,0)}]

		\drawrow{1} {28888888288888882222}
		\drawrow{2} {82888882828888828222}
		\drawrow{3} {88288828882888288822}
		\drawrow{4} {88828288888282888882}
		\drawrow{5} {88882888888828888888}
		\drawrow{6} {88828288888282888882}
		\drawrow{7} {88288828882888288822}
		\drawrow{8} {82888882828888828222}
		\drawrow{9} {28888888288888882222}
		\drawrow{10}{82888882828888828222}
		\drawrow{11}{88288828882888288822}
		\drawrow{12}{88828288888282888882}
		\drawrow{13}{88882888888828888888}
		\drawrow{14}{88828288888282888882}
		\drawrow{15}{88288828882888288822}
		\drawrow{16}{82888882828888828222}		
		\drawrow{17}{28888888288888882222}
		\drawrow{18}{22888882228888822222}
		\drawrow{19}{22288822222888222222}
		\drawrow{20}{22228222222282222222}

		\drawbox{20}{20};
	

	\endscope

	\end{tikzpicture} 
        \caption{4-diamonds packed in a \recdim{20}{20} board.\label{fig:packed}}
    \end{minipage}
\end{figure}

First consider $S_0$. At the start of Stage $3$, row $y(0)$ is entirely in the top-left region, so a single cycle of the $c$ colours suffices to expand the region to include row $y(0)-1$. Each subsequent cycle of $c$ colours expands the region to include an additional row. Therefore, after $c(\lceil q/2 \rceil -1)\leq cq/2$ moves of Stage $3$, all rows above $y(0)$ are included in the top left region. Similarly, the section $S_\ell$ will be included in the top-left region as it contains $\lfloor q/2 \rfloor \leq q/2$ rows.

Now consider section $S_i$ for some $0 < i < \ell$. Observe that there are at most $q$ rows in $S_i$ which are not already completely in the top-left section (after stage~2). Further observe that any cycle of $c$ colours expands the region to include \emph{two} more of these rows. One row is gained from the region bordering the top of the section (which is in the top-left region from stage $2$). The second is gained from the region bordering the top of the section (which is also in the top-left region from stage $2$). Therefore after at most $c\lceil{q/2}\rceil$ moves of Stage $3$ the board is flooded.

Over all three stages this gives a total of at most $n + \ell n + c\lceil{q/2}\rceil$ moves. We pick $\ell=\lceil{{\sqrt{c/2}}}\rceil$ to minimise this number of moves. By recalling that $q=\lfloor n/\ell \rfloor$ and simplifying we have that this total is less than $2n + {\sqrt{2c}\,n} + c$ moves as required. \qed
\end{proof}

\begin{theorem} \label{thm:bad}
    For $2\leq c\leq n^2$, there exists an \recdim{n}{n} board with (up to) $c$ colours which requires at least $\sqrt{c-1}\, n/2-c/2$ moves to flood.
\end{theorem}
\begin{proof}

    Suppose first that $c$ is even. For a given integer $r\geq 1$, let $D_{(x,y)}$ be an $r$-diamond where odd layers are coloured $x$ and even layers are coloured $y$. Any board containing $D_{(x,y)}$ requires at least $r$ moves of colours $x$ and $y$. Further, observe that as long as the centre of $D_{(x,y)}$ is in the board, even if it is cropped by at most two edges of the board, at least $r$ moves of colours $x$ and $y$ are still required (see Figure~\ref{fig:kdiamonds}b). We refer to such an $r$-diamond as \emph{good}. The central idea is to populate the board with good $r$-diamonds, $D_{(1,2)},D_{(3,4)}, \ldots, D_{(c-1,c)}$. As each $r$-diamond uses two colours (or one of the two colours if $r=1$) which do not occur in any other diamond, the board must take at least $rc/2$ moves to flood.

    It is not difficult to show that at least $(n^2-r^2)/(2r^2)$ good $r$-diamonds can be embedded in an \recdim{n}{n} board. An example of such a packing for a \recdim{20}{20} board is given in Figure~\ref{fig:packed} (which shows only the edges of diamonds and not their colouring). This scheme generalises well to an \recdim{n}{n} board but the details are omitted in the interest of brevity.
    

    We now take $r=\lfloor n/\sqrt{c}\rfloor < n/2$ and note that $r\geq 1$. As $r<n/2$, the $r$-diamonds are cropped by at most two board edges as required. Therefore we have at least $(n^2-r^2)/(2r^2) \geq c/2 - 1/2$ good $r$-diamonds in our board. However, as the number of good $r$-diamonds is an integer, this is at least $c/2$ as required. Therefore, the number of moves required to flood this board is at least $rc/2 > n\sqrt{c}/2 - c/2\,.$

    Finally, in the case that $c$ is odd we proceed as above using $c-1$ of the colours to give the stated result. \qed
\end{proof}

The next corollary is immediate from Theorems~\ref{thm:good} and~\ref{thm:bad}.

\begin{corollary}
    \label{col:moves}
    $({\sqrt{c-1}\, n-c})/{2} \leq \Mnc \leq 2n+\sqrt{2c}\,n+c$\,.
\end{corollary}

\section{Random boards} \label{sec:random}

In this section, we try to understand the complexity of a random Flood-It board -- that is, a board where each tile is coloured uniformly at random. This question is of both theoretical and practical interest. A common initialisation for Flood-It is to pick the colours of tiles at random and the game designer will surely be keen to know if they are likely to have chosen an instance whose solution is trivially short. The option of having to solve every created instance to test for this possibility is also likely to be unattractive, especially given the complexity results shown in this paper.  Intuitively, one would expect random boards to usually require a large number of moves to flood. Determining how many moves are actually needed turns out to be closely related to a body of research in {\em percolation theory}, the study of connected clusters in random graphs.

Indeed, a problem in percolation theory that is essentially equivalent to the question of the number of moves required for a random Flood-It board has been solved quite recently by Chayes and Winfield \cite{CW1993:random}, and independently Fontes and Newman \cite{FN1993:random}. In our terminology, their result was that a random \recdim{n}{n} Flood-It board with $c \geq 2$ colours requires $\Omega(n)$ moves with high probability. The proofs are lengthy and use some deep previous results in percolation theory.

We now present a greatly simplified proof of the results of \cite{CW1993:random,FN1993:random}, in the case that $c \geq 3$. Formally, our result is as follows.

\begin{theorem}
    \label{thm:random}
    Let $B \in B_{n,c}$ be a board where the colour of each tile is chosen uniformly at random from $\{1,\dots,c\}$. Then, for $c \geq 4$, $\Pr[m(B) \leq 2(3/10-1/c)(n-1)] < e^{-\Omega(n)}$. For $c=3$, $\Pr[m(B) \leq (n-1)/22] < e^{-\Omega(n)}$.
\end{theorem}

In order to prove this theorem, we will use two lemmas concerning paths in Flood-It boards. Let $P$ be a simple path in a Flood-It board, i.e.\ a simple path on the underlying square lattice\footnote{Simple paths on square lattices have been intensively studied, and are known as {\em self-avoiding walks} \cite{MS1996:SAW}. There are known upper bounds, which are slightly stronger than Lemma \ref{lem:pathcount}, on the number of self-avoiding walks of a given length; however, we avoid these here to keep our presentation elementary.}, where tiles are vertices on the path. Note that a path of length $k$ includes $k+1$ tiles. We say that a simple path $P$ is {\em non-touching} if every tile in $P$ is adjacent to at most two tiles that are also in $P$. Define the {\em cost} of $P$, $\text{cost}(P)$, to be the number of maximal monochromatic connected components of the path, minus one (so a monochromatic path has cost 0).

\begin{lemma}
    \label{lem:pathcost}
    For any $B \in B_{n,c}$, there is a non-touching path from $(1,1)$ to $(n,n)$ with cost at most $m(B)$.
\end{lemma}

\begin{proof}
    For $m(B)=0$ there is nothing to prove, so consider a strategy for completing $B$ which uses $m(B)>0$ moves. Label every tile $t \in B$ with an integer $m(t)$ between 0 and $m(B)$ that indicates the number of the move which changed the colour of $t$ to be the colour of tile $(1,1)$. Then, for each $i \geq 1$, there is a connected component labelled with $i$ which has at least one neighbour labelled with $i-1$. As the label of $(n,n)$ is at most $m(B)$, and the label of $(1,1)$ is 0, there is a simple path from $(1,1)$ to $(n,n)$ with cost at most $m(B)$. This path can be taken to be non-touching, because any pair of adjacent tiles $(t_1,t_2)$ that are on the path but not connected by it correspond to a loop in the path that can be removed without increasing the cost.
    \qed
\end{proof}

\begin{lemma}
    \label{lem:pathcount}
    For any integer $\ell\geq 3$, there are at most $4 \cdot 7^{(\ell-1)/2} < 2 \cdot (\sqrt{7})^\ell$ non-touching paths of length $\ell$ from any given tile.
\end{lemma}

\begin{proof}
    Let $T(\ell)$ denote the maximum number of non-touching paths of length $\ell$ starting from any given tile. $T(\ell)$ can be straightforwardly upper bounded by $4 \cdot 3^{\ell-1}$ for $\ell \geq 1$, as with each step of the path, aside from the first, there are at most 3 choices of direction. We get a tighter bound by analysing a few steps on a non-touching path $P$. Consider the $i$th vertex on $P$, for some $i \geq 2$. As $P$ is simple, there are at most 3 choices for the $(i+1)$th vertex of the path. For vertex $i+2$, if the previous two steps were in the same direction, there are at most 3 more choices. On the other hand, if the previous two were in different directions, there are only at most 2 choices (otherwise, the path would go back on itself, and would not be non-touching). In total, there are only at most 7 possible options for vertices $i+1$, $i+2$. Therefore, for any $\ell \geq 3$, we have $T(\ell) \leq 4 \cdot 7^{(\ell-1)/2}$.
    \qed
\end{proof}

\noindent The last result we will need is the following Chernoff-Hoeffding bound.

\begin{fact}[Hoeffding \cite{hoeffding63}]
    \label{fact:chernoff}
    Let $X_i$, $1 \leq i \leq m$, be independent 0/1-valued random variables with $\Pr[X_i=1]=p$ then, \[\Pr\left[\frac{1}{m}\sum_{i=1}^m X_i \geq p + \epsilon\right] \leq e^{-D(p + \epsilon||p) m} \leq e^{-2\epsilon^2 m}\,,\] where $D(x||y)$ is the Kullback-Leibler divergence $D(x||y) = x \ln (x/y) + (1-x)\ln((1-x)/(1-y))\,.$
\end{fact}

\noindent We are finally ready to prove Theorem \ref{thm:random}.

\begin{proof}[of Theorem \ref{thm:random}]
    For any $k \geq 0$, and for any board $B$ such that $m(B)\leq k$, by Lemma \ref{lem:pathcost} there exists a non-touching path from $(1,1)$ to $(n,n)$ with cost at most $k$. So consider an arbitrary non-touching path $P$ in $B$ of length $\ell$ between these two tiles, and let $P_i$ denote the $i$th tile on the path, for $1 \leq i \leq \ell + 1$. Note that $\ell\geq 2(n-1)$. Then $\cost(P)=|\{i:P_{i+1} \neq P_i\}|$, or equivalently $\cost(P)=\ell - |\{i:P_{i+1} = P_i\}|$. Define the 0/1-valued random variable $X_i$ by $X_i=1 \Leftrightarrow P_{i+1} = P_i$. Then, as the colours of tiles are uniformly distributed, $\Pr[X_i=1] = 1/c$ for all $i$, and
    \[ \Pr[\cost(P) \leq k] = \Pr\left[\sum_{i=1}^\ell X_i \geq \ell-k\right] \leq e^{-D(1-k/\ell\,||\,1/c)\ell}, \]
    where we use Fact \ref{fact:chernoff}. Thus, using the union bound over all paths of length at least $2(n-1)$ from $(1,1)$ to $(n,n)$, we get that the probability that there exists {\em any} path of cost at most $k$ is upper bounded by
    \begin{equation}
    \label{eqn:union}
    2 \sum_{\ell=2(n-1)}^{\infty} (\sqrt{7})^\ell e^{-D(1-k/\ell\,||\,1/c)\ell} = 2 \sum_{\ell=2(n-1)}^{\infty} e^{((1/2)\ln 7 - D(1-k/\ell\,||\,1/c))\ell},
    \end{equation}
    where we use the estimate for the number of paths which was derived in Lemma~\ref{lem:pathcount}. In the final part of the proof, we consider the cases $c \geq 4$ and $c=3$ separately.

    First suppose that $c\geq 4$. We take $k = 2(3/10-1/c)(n-1) \leq (3/10-1/c)\ell$, as in the statement of the theorem, and use $D(1-k/\ell\,||\,1/c) \geq 2(1-k/\ell-1/c)^2$ (from Fact~\ref{fact:chernoff}) to obtain the bound
    \[ 2 \sum_{\ell=2(n-1)}^{\infty} e^{((1/2)\ln 7 - 2(1-k/\ell-1/c)^2)\ell} \leq 2 \sum_{\ell=2(n-1)}^{\infty} e^{((1/2)\ln 7 - 49/50)\ell}. \]
    As $49/50 > (1/2)\ln 7 \approx 0.973$, this sum is exponentially small in $n$.

    Lastly, suppose that $c=3$. In this case, our choice of $k$ above is negative. Instead we take $k=(n-1)/22$, which implies $1-k/\ell\geq 43/44$. In order to obtain a sufficiently tight bound on $D(1-k/\ell\,||\,1/c)$, we use the explicit formula in Fact~\ref{fact:chernoff} to show that $D(43/44\,||\,1/3) > 0.974 > (1/2)\ln 7$, which implies that there is a bound in Equation~(\ref{eqn:union}) which is exponentially small in $n$. This completes the proof.
    \qed
\end{proof}

\section{Conclusion and open problems} \label{sec:conc}
We have shown that, for three or more colours, \Floodit{} is \NPtime-hard. However, for two colours, the relaxed version of the problem termed \ColoroidF{} in which we are allowed to flood fill from any tile of the board remains in \Ptime. Some interesting open questions remain. First, the complexity of solving \ColoroidF{} on a height 2 board remains unresolved. Second, we conjecture that the true lower bound for random boards is $\Omega(\sqrt{c}n)$, matching the general upper bound. Interestingly, the percolation theory techniques that we are aware of do not appear to allow for super-linear lower bounds of the sort that would be required.

\section{Acknowledgements}
AM was funded by an EPSRC Postdoctoral Research Fellowship. MJ was supported by the EPSRC. We are grateful to Dave Arthur for producing an implementation of Flood-It, complete with examples of our NP-hardness reductions.  We
would also like to thank Leon Atkins, Aram Harrow, Tom Hinton and Alex Popa for many helpful and encouraging discussions.

\bibliographystyle{plain}
\bibliography{flood-it}

\end{document}